\documentclass[10pt]{article}

\usepackage[hmargin=35mm,vmargin=40mm]{geometry}
\usepackage{amsmath}
\usepackage{amssymb}
\usepackage{amsthm}
\usepackage{graphicx}
\usepackage{url}
\usepackage[noend]{algpseudocode} 

\newtheorem{theorem}{Theorem}

\newtheorem{lemma}[theorem]{Lemma}
\newtheorem{problem}[theorem]{Problem}

\theoremstyle{definition}
\newtheorem{definition}[theorem]{Definition}
\newtheorem{algorithm}[theorem]{Algorithm}

\newcounter{c-save-theorem}

\newcommand{\ti}[1]{#1.\mathit{idx}}
\newcommand{\tx}[1]{#1.\mathit{x}}
\newcommand{\ty}[1]{#1.\mathit{y}}

\newcommand{\lf}[1]{#1.\mathit{first}}
\newcommand{\lb}[1]{#1.\mathit{last}}
\newcommand{\lz}[1]{#1.\mathit{size}}
\newcommand{\OO}{\mathcal{O}}

\newcommand{\vnext}{\mathit{next}}
\newcommand{\vinner}{\mathit{inner}}
\newcommand{\vouter}{\mathit{outer}}
\newcommand{\vsbest}{s_\mathrm{best}}
\newcommand{\vtbest}{t_\mathrm{best}}

\newcommand{\cpp}{C\nolinebreak\hspace{-.05em}%
    \raisebox{.4ex}{\small+}\nolinebreak\hspace{-.10em}%
    \raisebox{.4ex}{\small+}}

\newcommand{\srcurl}{http://www.maths.uq.edu.au/{\footnotesize$\sim$}bab/code/}

\begin{document}

\title{Locating regions in a sequence under density constraints}
\author{Benjamin A.\ Burton and Mathias Hiron}
\date{3 April, 2013} 

\maketitle

\begin{abstract}
    Several biological problems require the identification of
    regions in a sequence where some feature occurs within a target density
    range: examples including the location of GC-rich regions, identification
    of CpG islands, and sequence matching.
    Mathematically, this corresponds
    to searching a string of 0s and 1s for a substring whose relative
    proportion of 1s lies between given lower and upper bounds.
    We consider the algorithmic problem of locating the longest such substring,
    as well as other related problems (such as finding the shortest substring
    or a maximal set of disjoint substrings).
    For locating the longest such substring, we develop an algorithm
    that runs in $\OO(n)$ time, improving upon the previous best-known
    $\OO(n \log n)$ result.
    For the related problems we develop $\OO(n \log \log n)$ algorithms,
    again improving upon the best-known $\OO(n \log n)$ results.
    Practical testing verifies that our new algorithms
    enjoy significantly smaller time and memory footprints, and can process
    sequences that are orders of magnitude longer as a result.




\end{abstract}

\section{Introduction}

In this paper we develop fast algorithms to search a sequence
for regions in which a given feature appears within a certain density range.
Such problems are common in biological sequence analysis; examples include:
\begin{itemize}
    \item \emph{Locating GC-rich regions}, where G and C nucleotides appear
    with high frequency.  GC-richness correlates with
    factors such as gene density \cite{zoubak96-distribution},
    gene length \cite{duret95-longgenes},
    recombination rates \cite{fullerton01-recombination},
    codon usage \cite{sharp95-silence},
    and the increasing complexity of organisms
    \cite{bernardi00-isochores,hardison91-alphaglobin}.

    \item \emph{Locating CpG islands}, which have a high
    frequency of CpG dinucleotides.
    CpG islands are generally associated with promoters
    \cite{ioshikhes00-mapping,saxonov06-promoters},
    are useful landmarks for identifying
    genes \cite{larsen92-cpg}, and play a role in cancer research
    \cite{esteller02-tumor}.

    \item \emph{Sequence alignment}, where we seek
    regions in which multiple sequences have a high rate of matches
    \cite{wang03-segments}.
\end{itemize}

Further biological applications of such problems are outlined in
\cite{goldwasser05-maxdense} and \cite{lin02-length}.
Such problems also have applications in other fields,
such as cryptography \cite{boztas09-density} and
image processing \cite{greenberg03-heavydense}.

We represent a sequence as a string of 0s and 1s (where 1 indicates the
presence of the feature that we seek, and 0 indicates its absence).
For instance, when locating GC-rich regions we let 1 and 0 denote GC and
TA pairs respectively.
For any substring, we define its \emph{density} to be the
relative proportion of 1s (which is a fraction between 0 and 1).
Our \emph{density constraint} is the following:
given bounds $\theta_1$ and $\theta_2$ with $\theta_1 < \theta_2$,
we wish to locate substrings whose density lies between
$\theta_1$ and $\theta_2$ inclusive.

The specific values of the bounds $\theta_1,\theta_2$
depend on the particular application.
For instance, CpG islands can be divided into classes according to their
relationships with transcriptional start sites \cite{ioshikhes00-mapping},
and each class is found to have its own characteristic range of GC content.
Likewise, isochores in the human genome can be classified into five
families, each exhibiting different ranges of GC-richness
\cite{bernardi00-isochores,zoubak96-distribution}.

We consider three problems in this paper:
\begin{enumerate}
    \item[(a)] locating the \emph{longest} substring with density
    in the given range;
    \item[(b)] locating the \emph{shortest} substring with density
    in the given range, allowing optional constraints on
    the substring length;
    \item[(c)] locating a \emph{maximal cardinality set} of disjoint
    substrings whose densities all lie in the given range, again
    with optional length constraints.
\end{enumerate}

The prior state of the art for these problems is described
by Hsieh et~al.\ \cite{hsieh08-interval}, who present $\OO(n \log n)$ algorithms
in all three cases.  In this paper we improve the time complexities of these
problems to $\OO(n)$, $\OO(n \log \log n)$ and $\OO(n \log \log n)$
respectively.  In particular, our $\OO(n)$ algorithm for problem~(a)
has the fastest asymptotic complexity possible.

Experimental testing on human genomic data verifies that our new
algorithms run significantly faster and require considerably less
memory than the prior state of the art,
and can process sequences that are orders of magnitude longer
as a result.

Hsieh et~al.\ \cite{hsieh08-interval} consider a more general setting for
these problems:
instead of 0s and 1s they consider strings of real numbers (whereupon
``ratio of 1s'' becomes ``average value'').
In their setting they prove
a theoretical \emph{lower bound} of $\Omega(n \log n)$ time on all three
problems.  The key feature that allows us to
break through their lower bound in this paper is the discrete
(non-continuous) nature of structures such as DNA; in other words, the
ability to represent them as strings over a finite alphabet.

Many other problems related to feature density are studied in the literature.
Examples include \emph{maximising} density under a length constraint
\cite{goldwasser05-maxdense,greenberg03-heavydense,lin02-length},
finding \emph{all} substrings under range of density and length
constraints \cite{hsieh08-interval,huang94-dna},
finding the longest substring whose density matches a \emph{precise}
value \cite{boztas09-density,burton11-density},
and one-sided variants of our first problem with
a lower density bound $\theta_1$ but no upper density bound $\theta_2$
\cite{allison03-longest,burton11-density,chen05-optimal,
hsieh08-interval,wang03-segments}.

We devote the first half of this paper to our $\OO(n)$ algorithm for
locating the longest substring with density in a given range:
Section~\ref{s-framework} develops the mathematical framework,
and Sections~\ref{s-algm} and~\ref{s-impl}
describe the algorithm and present performance testing.
In Section~\ref{s-related} we adapt our techniques for the remaining
two problems.

All time and space complexities in this paper are based on the
commonly-used \emph{word RAM model} \cite[\S 2.2]{cormen09-algorithms},
which is reasonable for modern computers.
In essence, if $n$ is the input size, we assume
that each $(\log n)$-bit integer takes constant space
(it fits into a single \emph{word}) and that simple arithmetical operations
on $(\log n)$-bit integers (words) take constant time.


\section{Mathematical Framework} \label{s-framework}

We consider a string of digits $z_1,\ldots,z_n$, where each $z_i$ is
either 0 or 1.
The \emph{length} of a substring $z_a,\ldots,z_b$ is defined to be
$L(a,b) = b-a+1$ (the number of digits it contains), and the
\emph{density} of a substring $z_a,\ldots,z_b$ is defined to be
$D(a,b) = \sum_{i=a}^b z_i / L(a,b)$ (the relative proportion of 1s).
It is clear that the density always lies in the range
$0 \leq D(a,b) \leq 1$.

Our first problem is to find the maximum length substring whose density
lies in a given range.  Formally:

\begin{problem} \label{p-range}
    Given a string $z_1,\ldots,z_n$ as described above and two
    rational numbers $\theta_1 = c_1/d_1$ and
    $\theta_2 = c_2/d_2$, compute
    \[ \max_{1 \leq a \leq b \leq n} \left\{
        L(a,b)\,|\,\theta_1 \leq D(a,b) \leq \theta_2 \right\}. \]
    We assume that $0 < \theta_1 < \theta_2 < 1$,
    that $0 < d_1,d_2 \leq n$, and that
    $\gcd(c_1,d_1) = \gcd(c_2,d_2) = 1$.
\end{problem}

For example, if the input string is $1100010101$ (with $n=10$)
and the bounds are $\theta_1=1/4$ and $\theta_2=1/3$ then
the maximum length is $7$.  This is attained by the substring
$11\underline{0001010}1$ ($a=3$ and $b=9$), which has
density $D(3,9) = 2/7 \simeq 0.286$.

The additional assumptions in Problem~\ref{p-range}
are harmless.  If $\theta_1=0$ or $\theta_2=1$ then
the problem reduces to a one-sided bound, for which simpler
linear-time algorithms are already known
\cite{allison03-longest,chen05-optimal,wang03-segments}.
If $\theta_1=\theta_2$ then the problem reduces to matching a precise
density, for which a linear-time algorithm is also known
\cite{burton11-density}.
If some $\theta_i$ is irrational or if some $d_i > n$, we can adjust
$\theta_i$ to a nearby rational for which $d_i \leq n$ without
affecting the solution.

We consider two geometric representations, each of which describes
the string $z_1,\ldots,z_n$ as a path in two-dimensional space.
The first is the natural representation, defined as follows.

\begin{definition}
    Given a string $z_1,\ldots,z_n$ as described above,
    the \emph{natural representation} is the sequence of
    $n+1$ points $\mathbf{p}_0,\ldots,\mathbf{p}_n$ where
    each $\mathbf{p}_k$ has coordinates $(k,\ \sum_{i=1}^k z_k)$.
\end{definition}

\begin{figure}[htb]
    \centering
    \includegraphics[scale=0.9]{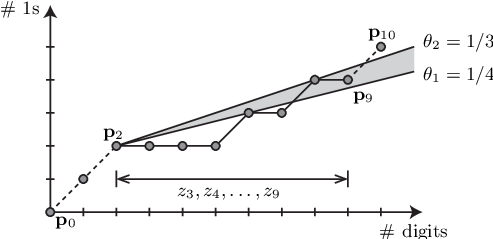}
    \caption{The natural representation of the string $1100010101$}
    \label{fig-plot-native}
\end{figure}

The $x$ and $y$ coordinates of $\mathbf{p}_k$
effectively measure the number of
digits and the number of 1s respectively in the prefix string
$z_1,\ldots,z_k$.  See Figure~\ref{fig-plot-native} for an illustration.

The natural representation is useful because densities have a clear
geometric interpretation:

\begin{lemma}
    In the natural representation, the density $D(a,b)$ is the gradient
    of the line segment joining $\mathbf{p}_{a-1}$ with $\mathbf{p}_b$.
\end{lemma}

The proof follows directly from the definition of $D(a,b)$.
The shaded cone in Figure~\ref{fig-plot-native} shows how,
for our example problem, the gradient of the line segment joining
$\mathbf{p}_2$ with $\mathbf{p}_9$ (i.e., the density $D(3,9)$)
lies within our target range $[\theta_1,\theta_2] = [1/4,1/3]$.

Our second geometric representation is the
\emph{orthogonal representation}.  Intuitively,
this is obtained by shearing the
previous diagram so that lines of slope $\theta_1$ and $\theta_2$
become horizontal and vertical respectively, as shown in
Figure~\ref{fig-plot-orth}.
Formally, we define it as follows.

\begin{figure}[htb]
    \centering
    \includegraphics[scale=0.9]{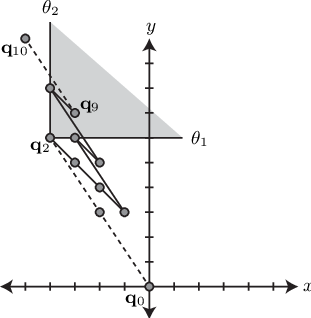}
    \caption{The orthogonal representation of the string $1100010101$
    for $[\theta_1,\theta_2] = [1/4,1/3]$}
    \label{fig-plot-orth}
\end{figure}

\begin{definition}
    Given a string $z_1,\ldots,z_n$ and rational numbers
    $\theta_1=c_1/d_1$ and $\theta_2=c_2/d_2$ as described earlier,
    the \emph{orthogonal representation} is the sequence of
    $n+1$ points $\mathbf{q}_0,\ldots,\mathbf{q}_n$, where
    each $\mathbf{q}_k$ has coordinates
    $( c_2 k - d_2 \sum_{i=1}^k z_i,
    \ -c_1 k + d_1 \sum_{i=1}^k z_i)$.
\end{definition}

From this definition we obtain the following immediate result.
\begin{lemma} \label{l-q-induct}
    $\mathbf{q}_0=(0,0)$, and for $i>0$ we have
    $\mathbf{q}_i = \mathbf{q}_{i-1} + (c_2,-c_1)$ if $z_i=0$ or
    $\mathbf{q}_i = \mathbf{q}_{i-1} + (c_2-d_2,d_1-c_1)$ if $z_i=1$.
\end{lemma}

The key advantage of the orthogonal representation is that
densities in the target range $[\theta_1,\theta_2]$
correspond to \emph{dominating points} in our new coordinate system.
Here we use a non-strict definition of domination:
a point $(x,y)$ is said to
\emph{dominate} $(x',y')$ if and only if both $x \geq x'$ and $y \geq y'$.

\begin{theorem} \label{t-density-dom}
    The density of the substring $z_a,\ldots,z_b$ satisfies
    $\theta_1 \leq D(a,b) \leq \theta_2$ if and only
    if $\mathbf{q}_b$ dominates $\mathbf{q}_{a-1}$.
\end{theorem}

\begin{proof}
    The difference $\mathbf{q}_b - \mathbf{q}_{a-1}$ has coordinates
    \begin{eqnarray*}
    & &
        ( c_2 (b-a+1) - d_2 \sum_{i=a}^b z_i,
        \ -c_1 (b-a+1) + d_1 \sum_{i=a}^b z_k) \nonumber \\
    &=&
        L(a,b) \cdot \left( c_2 - d_2 D(a,b),\ -c_1 + d_1 D(a,b) \right),
    \end{eqnarray*}
    which are both non-negative if and only if
    $D(a,b) \leq c_2/d_2 = \theta_2$ and
    $D(a,b) \geq c_1/d_1 = \theta_1$.
\end{proof}

The shaded cone in Figure~\ref{fig-plot-orth} shows how $\mathbf{q}_9$
dominates $\mathbf{q}_2$ in our example, indicating that the
substring $z_3,\ldots,z_9$ has a density in the range $[1/4,1/3]$.

It follows that Problem~\ref{p-range} can be reinterpreted as:

\setcounter{c-save-theorem}{\arabic{theorem}}
\renewcommand{\thetheorem}{\ref{p-range}$'$}
\begin{problem} \label{p-range-dom}
    Given the orthogonal representation $\mathbf{q}_0,\ldots,\mathbf{q}_n$
    as defined above, find points
    $\mathbf{q}_s,\mathbf{q}_t$ for which
    $\mathbf{q}_t$ dominates $\mathbf{q}_s$ and
    $t-s$ is as large as possible.
\end{problem}
\renewcommand{\thetheorem}{\arabic{theorem}}
\setcounter{theorem}{\arabic{c-save-theorem}}

The corresponding substring that solves Problem~\ref{p-range}
is $z_{s+1},\ldots,z_t$.

We finish this section with two properties of the orthogonal
representation that are key to obtaining a linear time algorithm
for this problem.

\begin{lemma} \label{l-int}
    The coordinates of each point $\mathbf{q}_i$ are integers in the
    range $[-n^2,n^2]$.
\end{lemma}

\begin{proof}
    This follows directly from Lemma~\ref{l-q-induct}:
    $\mathbf{q}_0 = (0,0)$, and the coordinates of each subsequent
    $\mathbf{q}_i$ are obtained by adding integers in the range
    $[-n,n]$ to the coordinates of $\mathbf{q}_{i-1}$.
\end{proof}

\begin{lemma} \label{l-moveout}
    If $\mathbf{q}_i$ dominates $\mathbf{q}_j$ then $i \geq j$.
\end{lemma}

\begin{proof}
    Consider the linear function
    $f \colon \thinspace \mathbb{R}^2 \to \mathbb{R}$
    defined by $f(x,y) = d_1 x + d_2 y$.  It is clear from
    Lemma~\ref{l-q-induct} that $f(\mathbf{q}_0) = 0$ and
    $f(\mathbf{q}_i) = f(\mathbf{q}_{i-1}) + d_1c_2 - d_2c_1$.
    Since $\theta_1 = c_1/d_1 < c_2/d_2 = \theta_2$ it follows that
    $f(\mathbf{q}_i) > f(\mathbf{q}_{i-1})$.

    Suppose $\mathbf{q}_i$ dominates $\mathbf{q}_j$.
    By definition of $f$ we have
    $f(\mathbf{q}_i) \geq f(\mathbf{q}_j)$, and by the observation
    above it follows that $i \geq j$.
\end{proof}


\section{Algorithm} \label{s-algm}

To solve Problem~\ref{p-range-dom} we construct and then scan along
the \emph{inner} and \emph{outer frontiers}, which we define as follows.

\begin{definition}
    Consider the orthogonal representation
    $\mathbf{q}_0,\ldots,\mathbf{q}_n$ for the input string
    $z_1,\ldots,z_n$.
    The \emph{inner frontier} is the set of points $\mathbf{q}_k$
    that do not dominate any $\mathbf{q}_i$ for $i \neq k$.
    The \emph{outer frontier} is the set of points $\mathbf{q}_k$
    that are not dominated by any $\mathbf{q}_i$ for $i \neq k$.
\end{definition}

\begin{figure}[htb]
    \centering
    \includegraphics[scale=0.9]{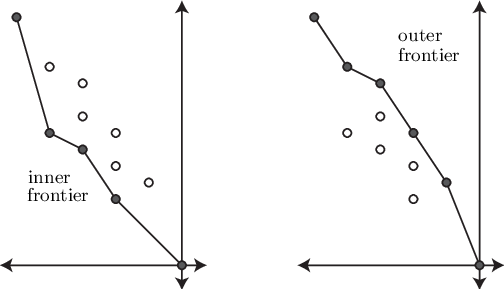}
    \caption{The inner and outer frontiers}
    \label{fig-frontier}
\end{figure}

Figure~\ref{fig-frontier} illustrates both of these sets.
They are algorithmically important because of the following result.

\begin{lemma}
    If $\mathbf{q}_s$ and $\mathbf{q}_t$ form a solution to
    Problem~\ref{p-range-dom}, then $\mathbf{q}_s$ lies on the inner
    frontier and $\mathbf{q}_t$ lies on the outer frontier.
\end{lemma}

\begin{proof}
    If $\mathbf{q}_s$ is not on the inner frontier then
    $\mathbf{q}_s$ dominates $\mathbf{q}_i$ for some $i \neq s$.
    By Lemma~\ref{l-moveout} we have $i < s$, which means that
    $\mathbf{q}_s$ and $\mathbf{q}_t$ cannot solve Problem~\ref{p-range-dom}
    since $\mathbf{q}_t$ dominates $\mathbf{q}_i$ and $t-i > t-s$.
    The argument for $\mathbf{q}_t$ on the outer frontier is similar.
\end{proof}


\subsection{Data structures}

The data structures that appear in this algorithm are simple.

For each point $\mathbf{q}_i=(x_i,y_i)$, we refer to $i$ as the
\emph{index} of $\mathbf{q}_i$, and we store the point as a
triple $(i,x_i,y_i)$.  If $t$ is such a triple, we refer to its
three constituents as $\ti{t}$, $\tx{t}$ and $\ty{t}$ respectively.

We make frequent use of \emph{lists} of triples.
If $L$ is such a list, we refer to the first and last triples in $L$
as $\lf{L}$ and $\lb{L}$ respectively, we denote the number of
triples in $L$ by $\lz{L}$, and we denote the individual
triples in $L$ by $L[0]$, $L[1]$, \ldots, $L[\lz{L}-1]$.
All lists are assumed to have $\OO(1)$ insertion and deletion at the
beginning and end, and $\OO(\lz{L})$ iteration through the elements in order
from first to last (as provided, for example, by a doubly-linked list).

For convenience we may write $\mathbf{q}_i \in L$ to indicate that
the triple describing $\mathbf{q}_i$ is contained in $L$; formally,
this means $(i,x_i,y_i) \in L$.


\subsection{The two-phase radix sort} \label{s-radix}

The algorithm make use of a \emph{two-phase radix sort} which,
given a list of $\ell$ integers in the range $[0,b^2)$,
allows us to sort these integers in $\OO(\ell+b)$ time and space.
In brief, the two-phase radix sort operates as follows.

Since the integers are in
the range $[0, b^2)$, we can express each integer $k$ as a
``two-digit number'' in base $b$; in other words, a pair $(\alpha,\beta)$
where $k=\alpha + \beta \cdot b$ and $\alpha,\beta$ are integers
in the range $0 \leq \alpha,\beta < b$.

We create an array of $b$ ``buckets'' (linked lists) in memory for each
possible value of $\alpha$.  In a first pass, we use a counting sort to
order the integers by increasing $\alpha$ (the least significant digit):
this involves looping through the integers to place each integer
in the bucket corresponding to the digit $\alpha$ (a total of
$\ell$ distinct $\OO(1)$ list insertion operations), and then
looping through the buckets to extract the integers in order of $\alpha$
(effectively concatenating $b$ distinct lists with total length $\ell$).
This first pass takes $\OO(\ell + b)$ time in total.

We then make a second pass using a similar approach, using another
$\OO(\ell + b)$ counting sort to order the integers
by increasing $\beta$ (the most significant digit).
Importantly, the counting sort is stable and so the final result has the
integers sorted by $\beta$ and then $\alpha$; that is, in numerical order.
The total running time is again $\OO(\ell + b)$, and since we have $b$
buckets with a total of $\ell$ elements, the space complexity is likewise
$\OO(\ell + b)$.

In our application, we need to sort
a list of $n+1$ integers in the range $[-n^2,n^2]$;
this can be translated into the setting above with $\ell=n+1$ and $b=2n$,
and so the two-phase radix sort has $\OO(n)$ time and space complexity.

This is a specific case of the more general radix sort; for further
details the reader is referred to a standard algorithms text such as
\cite{cormen09-algorithms}.


\subsection{Constructing frontiers}

The first stage in solving Problem~\ref{p-range-dom}
is to construct the inner and outer
frontiers in sorted order, which we do efficiently as follows.
The corresponding pseudocode is given in Figure~\ref{pc-frontier}.

\begin{algorithm} \label{a-frontier}
    To construct the inner frontier $I$ and the outer frontier $O$,
    both in order by increasing $x$ coordinate:
    \begin{enumerate}
        \item Build a list $L$ of triples corresponding to all
        $n+1$ points $\mathbf{q}_0,\ldots,\mathbf{q}_n$,
        using Lemma~\ref{l-q-induct}.
        Sort this list by increasing $x$ coordinate using a
        two-phase radix sort as described above,
        noting that the sort keys $x_i$ are all integers in the
        range $[-n^2,n^2]$ (Lemma~\ref{l-int}).

        \item \label{en-frontier-i}
        Initialise $I$ to the one-element list $[\lf{L}]$.
        Step through $L$ in forward order (from left to right
        in the diagram); for each triple
        $\ell \in L$ that has lower $y$ than any triple seen before,
        append $\ell$ to the end of $I$.

        \item \label{en-frontier-o}
        Construct $O$ in a similar fashion, working through $L$
        in reverse order (from right to left).
    \end{enumerate}

    In step~\ref{en-frontier-i}, there is a
    complication if we append a new triple $\ell$ to $I$
    for which $\tx{\ell} = \tx{\lb{I}}$.
    Here we must first remove $\lb{I}$ since $\ell$ makes it obsolete.
    See lines \ref{line-domi-start}--\ref{line-domi-end}
    of Figure~\ref{pc-frontier} for the details.

\begin{figure}[t]
    \centering
    \setlength{\fboxsep}{0.5\baselineskip}
    \framebox{\begin{minipage}{0.9\textwidth}
    \begin{algorithmic}[1]
        \State{$L \gets [(0,0,0)]$}
        \For{$i \gets 1$ \textbf{to} $n$}
            \If{$z_i = 0$}
                \State{Append $\lb{L} + (1,c_2,-c_1)$ to the end of $L$}
            \Else
                \State{Append $\lb{L} + (1,c_2-d_2,d_1-c_1)$ to the end of $L$}
            \EndIf
        \EndFor
        \State{Sort $L$ by increasing $x$ using a two-phase radix sort}

        \Statex

        \State $I \gets [ \lf{L} ]$
        \ForAll{$\ell \in L$, moving forward through $L$}
            \If{$\ty{\ell} < \ty{\lb{I}}$}
                \If{$\tx{\ell} = \tx{\lb{I}}$}
                        \Comment{$\lb{I}$ dominates $\ell$}
                        \label{line-domi-start}
                    \State{Remove the last triple from $I$}
                        \label{line-domi-end}
                \EndIf
                \State{Append $\ell$ to the end of $I$}
            \EndIf
        \EndFor

        \Statex

        \State $O \gets [ \lb{L} ]$
        \ForAll{$\ell \in L$, moving backwards through $L$}
            \If{$\ty{\ell} > \ty{\lf{O}}$}
                \If{$\tx{\ell} = \tx{\lf{O}}$}
                        \Comment{$\ell$ dominates $\lf{O}$}
                        \label{line-domo-start}
                    \State{Remove the first triple from $O$}
                        \label{line-domo-end}
                \EndIf
                \State{Prepend $\ell$ to the beginning of $O$}
            \EndIf
        \EndFor
    \end{algorithmic}
    \end{minipage}}
    \caption{The pseudocode for Algorithm~\ref{a-frontier}}
    \label{pc-frontier}
\end{figure}
\end{algorithm}

Table~\ref{tab-frontier-i} shows a worked example for
step~\ref{en-frontier-i} of the algorithm, i.e., the construction
of the inner frontier.  The points in this example correspond to
Figure~\ref{fig-frontier}, and each row of the table shows how the
frontier $I$ is updated when processing the next triple $\ell \in L$
(for simplicity we only show the coordinate pairs $(x_i,y_i)$ from each
triple).
Note that, although $L$ is sorted by increasing $x$ coordinate, for
each fixed $x$ coordinate the corresponding $y$ coordinates
may appear in arbitrary order.

\begin{table}\small
    \centering
    \begin{tabular}{l|l}
    Coordinates $(x_i, y_i)$ & Current inner frontier $I$ \\
    from the triple $\ell \in L$ & \\
    \hline
    \qquad$(-10,15)$ & $[\ (-10,15)\ ]$ \\
    \qquad$(-8,\phantom{0}12)$ & $[\ (-10,15),\ (-8,12)\ ]$ \\
    \qquad$(-8,\phantom{0}8)$ &  $[\ (-10,15)\ ]$ \\
                              &  $[\ (-10,15),\ (-8,8)\ ]$ \\
    \qquad$(-6,\phantom{0}9)$ & $[\ (-10,15),\ (-8,8)\ ]$ \\
    \qquad$(-6,\phantom{0}7)$ & $[\ (-10,15),\ (-8,8)\ (-6,7)\ ]$ \\
    \qquad$(-6,\phantom{0}11)$ & $[\ (-10,15),\ (-8,8)\ (-6,7)\ ]$ \\
    \qquad$(-4,\phantom{0}6)$ & $[\ (-10,15),\ (-8,8)\ (-6,7)\ (-4,6)\ ]$ \\
    \qquad$(-4,\phantom{0}8)$ & $[\ (-10,15),\ (-8,8)\ (-6,7)\ (-4,6)\ ]$\\
    \qquad$(-4,\phantom{0}4)$ & $[\ (-10,15),\ (-8,8)\ (-6,7)\ ]$\\
                              & $[\ (-10,15),\ (-8,8)\ (-6,7)\ (-4,4)\ ]$\\
    \qquad$(-2,\phantom{0}5)$ & $[\ (-10,15),\ (-8,8)\ (-6,7)\ (-4,4)\ ]$\\
    \qquad$(-0,\phantom{0}0)$ & $[\ (-10,15),\ (-8,8)\ (-6,7)\ (-4,4)\ (0,0)\ ]$
    \end{tabular}
    \caption{Constructing the inner frontier}
    \label{tab-frontier-i}
\end{table}

\begin{theorem} \label{t-frontier}
    Algorithm~\ref{a-frontier} constructs the inner and outer
    frontiers in $I$ and $O$ respectively, with each list sorted
    by increasing $x$ coordinate, in $\OO(n)$ time and $\OO(n)$ space.
\end{theorem}

\begin{proof}
    This algorithm is based on a well-known method for constructing frontiers.
    We show here why the inner frontier $I$ is constructed correctly;
    a similar argument applies to the outer frontier $O$.

    If a triple $\ell \in L$ with coordinates $(x_i,y_i)$ does belong on
    the inner
    frontier (i.e., there is no other point $(x_j,y_j)$ in the list that it
    dominates),
    then we are guaranteed to add it to $I$ in step~\ref{en-frontier-i}
    because the only triples processed thus far have $x \leq x_i$, and
    must therefore have $y > y_i$.  Moreover, we will not subsequently
    remove $\ell$ from $I$ again,
    since the only other triples with the same $x$
    coordinate must have $y > y_i$.

    If a triple $\ell \in L$ with coordinates $(x_i,y_i)$ does not belong
    on the inner frontier, then there is some point $(x_j,y_j)$ that it
    dominates.  If $x_j<x_i$ then we never add $\ell$ to $I$, since by
    the time we process $\ell$ we will
    already have seen the coordinate $y_j$ (which is at least as low as $y_i$).
    Otherwise $x_j=x_i$ and $y_j<y_i$, and so we either add and then
    remove $\ell$ from $I$ or else never add it at all,
    depending on the order in which we process
    the points with $x$ coordinate equal to $x_i$.
    Either way, $\ell$ does not appear in the final list $I$.

    Therefore the list $I$ contains precisely the inner frontier; moreover,
    since we process the points by increasing $x$ coordinate, the list $I$
    will be sorted by $x$ accordingly.

    The main innovation in this algorithm
    is the use of a radix sort with two-digit keys, made possible by
    Lemma~\ref{l-int}, which allow us to avoid the usual $\OO(n \log n)$
    cost of sorting.
    The two-phase radix sort runs in $\OO(n)$ time and space,
    as do the subsequent
    list operations in Steps~2--4, and so the entire algorithm runs in
    $\OO(n)$ time and space as claimed.
\end{proof}


\subsection{Sliding windows}

The second stage in solving Problem~\ref{p-range-dom} is to
simultaneously scan through the inner and outer frontiers in search of
possible solutions $\mathbf{q}_s \in I$ and $\mathbf{q}_t \in O$.

We do this by trying each $\mathbf{q}_s$ in order on the inner frontier,
and maintaining a sliding window $W$ of possible points
$\mathbf{q}_t$; specifically, $W$ consists of all points on the
outer frontier that dominate $\mathbf{q}_s$.
Figure~\ref{fig-windows} illustrates this window $W$ as
$\mathbf{q}_s$ moves along the inner frontier from left to right.

\begin{figure}[htb]
    \centering
    \includegraphics[scale=0.85]{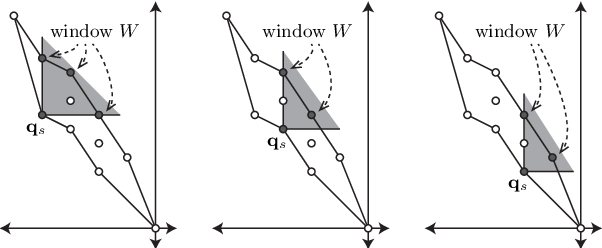}
    \caption{A sequence of sliding windows on the outer frontier}
    \label{fig-windows}
\end{figure}

For each point $\mathbf{q}_s \in I$ that we process,
it is easy to update $W$ in amortised $\OO(1)$ time using
sliding window techniques (pushing new points onto the end of the list
$W$ as they enter the window, and removing old points from the beginning
as they exit the window).
However, we still need a fast way of locating
the point $\mathbf{q}_t \in O$
that dominates $\mathbf{q}_s$ and for which $t-s$ is largest.
Equivalently, we need a fast way of choosing the triple $w \in W$ that
maximises $\ti{w}$.

To do this, we maintain a sub-list $M \subseteq W$: this is a list
consisting of all
triples $w \in W$ that are \emph{potential maxima}.  Specifically,
for any triple $w \in W$, we include $w$ in $M$ if and only if there is
no $w'\in W$ for which $\tx{w'} > \tx{w}$ and $\ti{w'} > \ti{w}$.
The rationale is that, if there were such a $w'$, we would always choose
$w'$ over $w$ in this or any subsequent window.

As with all of our lists, we keep $M$ sorted by increasing $x$
coordinate.  Note that the condition above implies that $M$ is also
sorted by \emph{decreasing} index.  In particular, the sought-after
triple $w \in W$ that maximises $\ti{w}$ is simply $\lf{M}$,
which we can access in $\OO(1)$ time.

Crucially, we can also update the sub-list $M$ in amortised $\OO(1)$ time for
each point $\mathbf{q}_s \in I$ that we process.  As a result,
this sub-list $M$ allows us to maximise $\ti{w}$ for $w \in W$ whilst
avoiding a costly linear scan through the entire window $W$.

The details are as follows; see Figure~\ref{pc-window} for the
pseudocode.

\begin{algorithm} \label{a-window}
    Let the inner and outer frontiers be stored in the lists
    $I$ and $O$ in order by increasing $x$ coordinate,
    as generated by Algorithm~\ref{a-frontier}.
    We solve Problem~\ref{p-range-dom} as follows.

    \begin{enumerate}
        \item
        Initialise $M$ to the empty list.

        \item Step through the inner frontier $I$ in forward order.
        For each triple $\vinner \in I$:
        \begin{enumerate}
            \item \emph{Process new points that enter our sliding window.}
            To do this, we scan through any new triples $\vouter \in O$
            for which $\ty{\vouter} \geq \ty{\vinner}$ and update
            $M$ accordingly.

            \smallskip

            \noindent
            Each new $\vouter \in O$ that we process has
            $\tx{\vouter} > \tx{m}$ for all $m \in M$, so we append $\vouter$
            to the end of $M$.
            However, before doing this we must remove any $m \in M$ for
            which $\ti{m} < \ti{\vouter}$ (since such triples would violate
            the definition of $M$).  Because $M$ is sorted by decreasing
            index, all such $m \in M$ can be found at the end of $M$.
            See lines \ref{line-enter-start}--\ref{line-enter-end}
            of Figure~\ref{pc-window}.

            \item \emph{Remove points from $M$ that have exited our
            sliding window.}  That is, remove triples $m \in M$
            for which $\tx{m} < \tx{\vinner}$.

            \smallskip

            \noindent
            Because $M$ is sorted by increasing $x$ coordinate,
            all such triples can be found at the beginning of $M$.
            See lines \ref{line-exit-start}--\ref{line-exit-end}
            of Figure~\ref{pc-window}.

            \item \emph{Update the solution.}
            The best solution to Problem~\ref{p-range-dom}
            that uses the triple $\vinner \in I$ is the
            pair of points
            $\mathbf{q}_{\ti{\vinner}},\allowbreak \mathbf{q}_{\ti{\lf{M}}}$.
            If the difference $\ti{\lf{M}}-\ti{\vinner}$ exceeds any
            seen so far, record this as the new best solution.
        \end{enumerate}
    \end{enumerate}
\end{algorithm}

\begin{figure}[t]
    \centering
    \setlength{\fboxsep}{0.5\baselineskip}
    \framebox{\begin{minipage}{0.9\textwidth}
    \begin{algorithmic}[1]
        \State $M \gets [\,]$
            \Comment{empty list}
        \State{$\vsbest \gets 0$, $\vtbest \gets 0$}
            \Comment{best solution so far}
        \State{$\vnext \gets 0$}
            \Comment{next element of $O$ to scan through}

        \Statex

        \ForAll{$\vinner \in I$, moving forward through $I$}
            \While{$\vnext < \lz{O}$ and
                    $\ty{O[\vnext]} \geq \ty{\vinner}$}
                    \label{line-enter-start}
                \State{$\vnext \gets \vnext + 1$}
                \While{$\lz{M} > 0$ and $\ti{O[\vnext]} > \ti{\lb{M}}$}
                        \label{line-redundant-start}
                    \State{Remove the last triple from $M$}
                    \EndWhile
                \State{Append $O[\vnext]$ to the end of $M$}
            \EndWhile \label{line-enter-end}
            \While{$\lz{M} > 0$ and $\tx{\lf{M}} < \tx{\vinner}$}
                    \label{line-exit-start}
                \State{Remove the first triple from $M$}
            \EndWhile \label{line-exit-end}
            \If{$\ti{\lf{M}} - \ti{\vinner} > \vtbest - \vsbest$}
                \State{$\vsbest \gets \ti{\vinner}$}
                \State{$\vtbest \gets \ti{\lf{M}}$}
            \EndIf
        \EndFor
    \end{algorithmic}
    \end{minipage}}
    \caption{The pseudocode for Algorithm~\ref{a-window}}
    \label{pc-window}
\end{figure}

\begin{theorem} \label{t-range-linear}
    Algorithms~\ref{a-frontier}
    and \ref{a-window} together solve Problem~\ref{p-range-dom}
    in $\OO(n)$ time and $\OO(n)$ space.
\end{theorem}

\begin{proof}
    Theorem~\ref{t-frontier} analyses Algorithm~\ref{a-frontier},
    and the preceding discussion shows the correctness of
    Algorithm~\ref{a-window}.  All that remains is to verify that
    Algorithm~\ref{a-window} runs in $\OO(n)$ time and space.

    Each triple $t \in O$ is added to $M$ at most once and removed from
    $M$ at most once, and so the \textbf{while} loops on lines
    \ref{line-enter-start}, \ref{line-redundant-start} and
    \ref{line-exit-start}
    each require \emph{total} $\OO(n)$ time as measured
    across the entire algorithm.
    Finally, the outermost \textbf{for} loop (line~4) iterates at most
    $n+1$ times, giving an overall running time for
    Algorithm~\ref{a-window} of $\OO(n)$.

    Each of the lists $I$, $O$ and $M$ contains at most $n+1$ elements,
    and so the space complexity is $\OO(n)$ also.
\end{proof}


\section{Performance} \label{s-impl}

Here we experimentally compare our new algorithm
against the
prior state of the art, namely the $\OO(n \log n)$ algorithm of
Hsieh et~al.\ \cite{hsieh08-interval}.
Our trials involve searching for GC-rich regions in the human genome assembly
\emph{GRCh37.p2} from \emph{GenBank} \cite{benson08-genbank,ihgsc04-human}.
The implementation that we use for our new algorithm is available online,%
\footnote{For {\cpp} implementations of all algorithms in this paper,
    visit \srcurl.}
and the code for the prior $\OO(n\log n)$ algorithm was downloaded from the
respective authors' website.\footnote{%
    The implementation of the prior algorithm \cite{hsieh08-interval}
    is taken from \url{http://venus.cs.nthu.edu.tw/~eric/FIF.htm}.}
Both implementations are written in C/{\cpp}.

\begin{figure}[htb]
    \centering
    \includegraphics[scale=0.5]{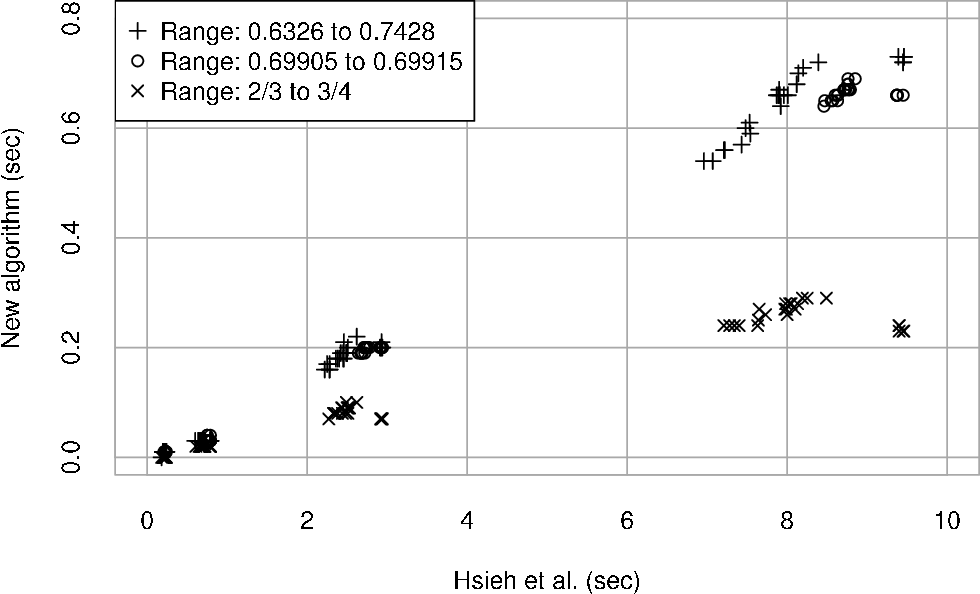}
    \caption{Performance comparisons on genomic data}
    \label{fig-longest-cmp}
\end{figure}

Figure~\ref{fig-longest-cmp} measures running times for
$24 \times 4 \times 3 = 288$ instances of Problem~\ref{p-range}:
we begin with 24 human chromosomes (1--22, X and Y),
extract initial strings of four different lengths $n$
(ranging from $n=100\,000$ to $n=3\,000\,000$),
and search each for the longest substring whose GC-density
is constrained according to one of three different ranges
$[\theta_1,\theta_2]$.

These ranges are:
$[0.6326,\,0.7428]$, which matches the first CpG island class of
Ioshikhes and Zhang \cite{ioshikhes00-mapping};
$[0.69905,\,0.69915]$, which surrounds the median of this
class and measures performance for a narrow density range;
and $[2/3,\,3/4]$, which measures performance when the key
parameters $d_1$ and $d_2$ are very small.

The results are extremely pleasing: in every case the new algorithm runs
at least $10 \times$ faster than the prior state of the art, and
in some cases up to $42 \times$ faster.
Of course such comparisons cannot be exact or fair, since the two
implementations
are written by different authors; however, they do illustrate that the
new algorithm is not just asymptotically faster in theory (as proven in
Theorem~\ref{t-range-linear}), but also fast in practice
(i.e., the constants are not so large as to eliminate the
theoretical benefits for reasonable inputs).

The results for the range $[2/3,\,3/4]$ highlight
how our algorithm benefits from small denominators
(in which the range of possible $x$ and $y$ coordinates becomes much smaller).

Memory becomes a significant problem when dealing with very large data sets.
The algorithm of Hsieh et~al.\ \cite{hsieh08-interval} uses ``heavy'' data
structures with large memory requirements:
for $n=3\,000\,000$ it uses $1.64$\,GB of memory.
In contrast, our new algorithm has a much smaller
footprint---just $70$\,MB for the same $n$---and can thereby
process values of $n$ that are orders of magnitude larger.
In Figure~\ref{fig-longest-human} we run our algorithm over
the full length of each chromosome; even the worst case
(chromosome~1) with $n=249\,250\,621$ runs for all density ranges
in under 85 seconds, using $5.6$\,GB of memory.

\begin{figure}[htb]
    \centering
    \includegraphics[scale=0.5]{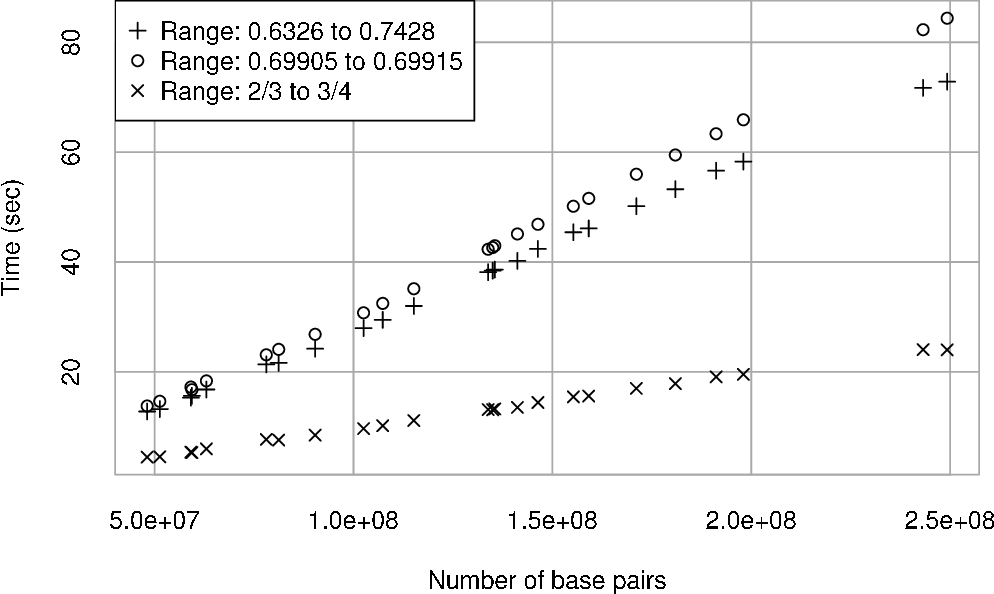}
    \caption{Performance of the new algorithm on full-length chromosomes}
    \label{fig-longest-human}
\end{figure}

All trials were run on a single 3\,GHz Intel Core i7 CPU.
Input and output are included in running times, though detailed
measurements show this to be insignificant for both algorithms
(which share the same input and output routines).


\section{Related Problems} \label{s-related}

The techniques described in this paper extend beyond Problem~\ref{p-range}.
Here we examine two related problems from the
bioinformatics literature, and for each we
outline new algorithms that improve upon the prior state of the art.
As usual, all problems take an input string $z_1,\ldots,z_n$ where
each $z_i$ is 0 or 1.

The new algorithms in this section rely on
\emph{van Emde Boas trees} \cite{vanemdeboas77-pqueue},
a tree-based data structure for which many elementary key-based
operations have $\OO(n \log \log n)$ time complexity.
We briefly review this data structure before presenting the two
related problems and the new algorithms to solve them.


\subsection{van Emde Boas trees} \label{s-veb}

Here we briefly recall the essential ideas behind van Emde Boas trees.
For full details we refer the author to a modern textbook
on algorithms such as \cite{cormen09-algorithms}.

A \emph{van Emde Boas tree} is a data structure that implements an associative
array (mapping keys to values), in which the user can perform several
elementary operations in $\OO(\log m)$ time, where $m$ is the number of
\emph{bits} in the key.
These elementary operations include inserting or deleting a
key-value pair, looking up the value stored for a given key,
and looking up the successor or predecessor of a given key $k$
(i.e., the first key higher or lower than $k$ respectively).
In our case, all keys are in the range $[-n^2,n^2]$ (Lemma~\ref{l-int}),
and so $m \in \OO(\log n^2) = \OO(\log n)$;
that is, these elementary operations run in
$\OO(\log \log n)$ time.

The core idea of this data structure is that each node of the tree represents a
range of $p$ consecutive possible keys for some $p$, and has $\sqrt{p}$
children (each a smaller van Emde Boas tree) that each represent a sub-range
of $\sqrt{p}$ possible keys.  Each node also maintains the minimum and maximum
keys that are actually present within its range.
The root node of the tree represents the complete range of $2^m$ possible keys.

Furthermore, for each node $V$ of the tree representing a range of $p$
possible keys, we also maintain an \emph{auxiliary} van Emde Boas tree
that stores which of the $\sqrt{p}$ children of $V$ are non-empty (i.e.,
have at least one key stored within them).

To look up the successor of a given key $k$ we travel down the tree,
and each time we reach some node $V_i$ that represents $p_i$ potential keys,
we identify which of the $\sqrt{p_i}$ children contains the successor of $k$
by examining the auxiliary tree attached to $V_i$.
This induces a query on the auxiliary tree (representing $\sqrt{p_i}$
potential non-empty child trees) followed by a query on the selected child
(representing $\sqrt{p_i}$ potential keys), and so the running time
follows a recurrence of the form $T(m) = 2T(m/2) + O(1)$;
solving this recurrence yields the overall $\OO(log(m))$ time complexity.
The running times for inserting and deleting keys follow a similar
argument, and again we refer the reader to a text such as
\cite{cormen09-algorithms} for the details.


\subsection{Shortest substring in a density range} \label{s-shortest}

The first related problem that we consider is a natural counterpart to
Problem~\ref{p-range}: instead of searching for the longest substring
under given density constraints, we search for the \emph{shortest}.

\begin{problem} \label{p-shortest}
    Find the shortest substring whose density lies in a given
    range.  That is, given rationals $\theta_1 < \theta_2$, compute
    \[ \min_{1 \leq a \leq b \leq n} \left\{
        L(a,b)\,|\,\theta_1 \leq D(a,b) \leq \theta_2 \right\}. \]
\end{problem}

The best known algorithm for this problem runs in $\OO(n \log n)$ time
\cite{hsieh08-interval}; here we improve this to $\OO(n \log \log n)$.

By Theorem~\ref{t-density-dom} and Lemma~\ref{l-moveout},
this is equivalent to finding points $\mathbf{q}_s \neq \mathbf{q}_t$
for which $\mathbf{q}_t$ dominates $\mathbf{q}_s$ and for which $t-s$ is as
\emph{small} as possible.  To do this, we iterate through each possible
endpoint $\mathbf{q}_t$ in turn, and maintain a \emph{partial outer frontier}
$P$ consisting of all non-dominated points amongst the previous points
$\{\mathbf{q}_0,\mathbf{q}_1,\ldots,\mathbf{q}_{t-1}\}$;
that is, all points $\mathbf{q}_i$ ($0 \leq i \leq t-1$)
that are not dominated by some other
$\mathbf{q}_j$ ($i < j \leq t-1$).
When examining a candidate endpoint $\mathbf{q}_t$, it is
straightforward to show that any optimal solution
$\mathbf{q}_s,\mathbf{q}_t$ must satisfy $\mathbf{q}_s \in P$.

\begin{algorithm} \label{a-shortest}
    To solve Problem~\ref{p-shortest}:
    \begin{enumerate}
        \item Initialise $P$ to the empty list.
        \item For each $t=0,\ldots,n$ in turn, try $\mathbf{q}_t$ as a
        possible endpoint:
        \begin{enumerate}
            \item \emph{Update the solution.}
            To do this, walk through all points $\mathbf{q}_s \in P$ that
            are dominated by $\mathbf{q}_t$.  If the difference
            $t-s$ is smaller than any seen so far, record this as the
            new best solution.

            \item \emph{Update the partial frontier}.
            To do this, remove all $\mathbf{q}_s \in P$
            that are dominated by $\mathbf{q}_t$, and then insert
            $\mathbf{q}_t$ into $P$.
        \end{enumerate}
    \end{enumerate}
\end{algorithm}

The key to a fast time complexity is choosing an
efficient data structure for storing the partial outer frontier $P$.
We keep $P$ sorted by increasing $x$ coordinate,
and for the underlying data structure
we use a \emph{van Emde Boas tree} \cite{vanemdeboas77-pqueue}.

To walk through all points $\mathbf{q}_s \in P$ that are dominated by
$\mathbf{q}_t$ in step~2a, we locate the point $\mathbf{p} \in P$
with smallest $x$ coordinate larger than $x_t$; a van Emde Boas tree
can do this in $\OO(\log \log n)$ time.  The dominated
points $\mathbf{q}_s$ can then be found immediately prior to
$\mathbf{p}$ in the partial frontier.
Adding and removing points in step~2b is likewise $\OO(\log \log n)$ time,
and the overall space complexity of a van Emde Boas tree can be made $\OO(n)$
\cite{cormen09-algorithms}.

A core requirement of this data structure is that keys in the tree can
be described by integers in the range $0,\ldots,n$.  To arrange this, we
pre-sort the $x$ coordinates
$\tx{\mathbf{q}_0},\ldots,\tx{\mathbf{q}_n}$ using an $\OO(n)$ two-phase
radix sort (as described in Section~\ref{s-radix}) and then replace each
$x$ coordinate with its corresponding rank.

To finalise the time and space complexities, we observe that
each point $\mathbf{q}_s \in P$ that is processed
in step~2a is immediately removed in step~2b, and so no point is processed more
than once.  Combined with the preceding discussion, this gives:

\begin{theorem} \label{t-shortest}
    Algorithm~\ref{a-shortest} runs in $\OO(n \log \log n)$ time
    and uses $\OO(n)$ space.
\end{theorem}

We can add an optional \emph{length constraint} to
Problem~\ref{p-shortest}: given rationals $\theta_1 < \theta_2$ and
length bounds $L_1 < L_2$, find the shortest substring $z_a,\ldots,z_b$
for which $\theta_1 \leq D(a,b) \leq \theta_2$ and
$L_1 \leq L(a,b) \leq L_2$.  This is a simple modification to
Algorithm~\ref{a-shortest}: we redefine the partial frontier
$P$ to be the set of all non-dominated points amongst
$\{\mathbf{q}_0,\mathbf{q}_1,\ldots,\mathbf{q}_{t-L_1}\}$.
The update procedure changes slightly, but the
$\OO(n \log \log n)$ running time remains.


\subsection{Maximal collection of substrings in a density range}

The second related problem involves searching for substrings ``in
bulk'': instead of finding the longest substring under given density
constraints, we find the \emph{most} disjoint substrings.

\begin{problem} \label{p-multiple}
    Find a maximum cardinality set of disjoint substrings
    whose densities all lie in a given range.
    That is, given rationals $\theta_1 < \theta_2$, find substrings
    $(z_{a_1},\ldots,z_{b_1})$,
    $(z_{a_2},\ldots,z_{b_2})$,
    \ldots,
    $(z_{a_k},\ldots,z_{b_k})$
    where $\theta_1 \leq D(a_i,b_i) \leq \theta_2$ and
    $b_i < a_{i+1}$ for each $i$,
    and where $k$ is as large as possible.
\end{problem}

As before, the best known algorithm runs in
$\OO(n \log n)$ time \cite{hsieh08-interval};
again we improve this bound to $\OO(n \log \log n)$.

For this problem we mirror the greedy approach of
Hsieh et~al.\ \cite{hsieh08-interval}.
One can show that, if $z_a,\ldots,z_b$ is a substring
of density $\theta_1 \leq D(a,b) \leq \theta_2$ with \emph{minimum
endpoint} $b$, then some optimal solution to Problem~\ref{p-multiple}
has $b_1=b$ (i.e., we can choose $z_a,\ldots,z_b$ as our first
substring).  See \cite[Lemma~6]{hsieh08-interval}.

Our strategy is to use our previous
Algorithm~\ref{a-shortest} to locate such a substring $z_a,\ldots,z_b$,
store this as part of our solution, and then rerun our algorithm on the
leftover $n-b$ input digits $z_{b+1},\ldots,z_n$.  We repeat this process
until no suitable substring can be found.

\begin{algorithm} \label{a-multiple}
    To solve Problem~\ref{p-multiple}:
    \begin{enumerate}
        \item Initialise $i \gets 1$.

        \item \label{en-leftover}
        Run Algorithm~\ref{a-shortest} on the input string
        $z_i,\ldots,z_n$, but terminate Algorithm~\ref{a-shortest}
        as soon as \emph{any}
        dominating pair $\mathbf{q}_s,\mathbf{q}_t$ is found.

        \item If such a pair is found: add the corresponding
        substring $z_{s+1},\ldots,z_t$ to our solution set,
        set $i \gets t+1$, and return to step~\ref{en-leftover}.
        Otherwise terminate this algorithm.
    \end{enumerate}
\end{algorithm}

It is important to reuse the same van Emde Boas tree on each run through
Algorithm~\ref{a-shortest} (simply empty out the tree each time),
so that the total initialisation cost remains $\OO(n\log\log n)$.

\begin{theorem}
    Algorithm~\ref{a-multiple} runs in $\OO(n \log \log n)$ time
    and requires $\OO(n)$ space.
\end{theorem}

\begin{proof}
    Running Algorithm~\ref{a-shortest} in step~\ref{en-leftover}
    takes $\OO([t-i] \allowbreak \log \log n)$ time,
    since it only examines points
    $\mathbf{q}_i,\ldots,\mathbf{q}_t$ before the dominating pair is found.
    The total running time of Algorithm~\ref{a-multiple} is therefore
    $\OO(b_1 \log \log n + [b_2-b_1] \log \log n + \ldots +
    [b_k-b_{k-1}] \log \log n) = \OO(n \log \log n)$.
    The $\OO(n)$ space complexity follows from
    Theorem~\ref{t-shortest} plus the observation that the final
    solution set can contain at most $n$ disjoint substrings.
\end{proof}

As before, it is simple to add a length constraint to
Problem~\ref{p-multiple}, so that each substring
$z_{a_i},\ldots,z_{b_i}$ must satisfy both
$\theta_1 \leq D(a_i,b_i) \leq \theta_2$ and
$L_1 \leq L(a_i,b_i) \leq L_2$ for some length bounds $L_1 < L_2$.
We simply incorporate the length constraint into
Algorithm~\ref{a-shortest} as described in Section~\ref{s-shortest},
and nothing else needs to change.


\subsection{Performance}

As before, we have implemented both Algorithms~\ref{a-shortest}
and~\ref{a-multiple}
and tested each against the prior
state of the art using human genomic data, following the same
procedures as described in Section~\ref{s-impl}.
Our van Emde Boas tree implementation is based on
the MIT-licensed \emph{libveb} by Jani Lahtinen, available from
http://code.google.com/p/libveb/.

For the shortest substring problem, our algorithm runs
at 12--62 times the speed of the prior algorithm
\cite{hsieh08-interval}, and requires
under $1/50$th of the memory.
For finding a maximal collection of substrings,
our algorithm runs at 0.94--13 times the speed of
the prior algorithm
\cite{hsieh08-interval}, and uses less than $1/20$th of the memory.%
    \footnote{The few cases in which our algorithm was slightly slower
    (down to $0.94$ times the speed) all involved large denominators $d_1,d_2$
    and maximal collections involving a very large number of
    very short substrings.}

Once again, these improvements---particularly for memory usage---allow
us to run our algorithms with significantly larger values of $n$
than for the prior state of the art.
For chromosome~1 with $n=249\,250\,621$, the two algorithms run in
221 and 176 seconds respectively, both using approximately 5.6\,GB of
memory.


\section{Discussion} \label{s-discussion}

In this paper we consider three problems involving the identification
of regions in a sequence where some feature occurs within a given density
range.  For all three problems we develop new algorithms that offer
significant performance improvements over the prior state of the art.
Such improvements are critical for disciplines such as bioinformatics
that work with extremely large data sets.

The key to these new algorithms is the ability to exploit
\emph{discreteness} in the input data.  All of the applications we
consider (GC-rich regions, CpG islands and sequence alignment) can
be framed in terms of sequences of 1s and 0s.  Such discrete
representations are powerful: through Lemma~\ref{l-int},
they allow us to perform $\OO(n)$ two-phase radix sorts, and
to reindex coordinates by rank for van Emde Boas trees.
In this way, discreteness allows us to circumvent the
theoretical $\Omega(n \log n)$ bounds of
Hsieh et~al.\ \cite{hsieh08-interval},
which are only proven for the more general continuous (non-discrete) setting.

In a discrete setting, an obvious lower bound for all three problems
is $\Omega(n)$ time (which is required to read the
input sequence).  For the first problem (longest substring with density
in a given range), we attain this best possible lower bound with our
$\OO(n)$ algorithm.
This partially answers a question of Chen and Chao \cite{chen05-optimal},
who ask in a more general setting whether such an algorithm is possible.

For the second and third problems (shortest substring and maximal
collection of substrings), although our $\OO(n \log \log n)$ algorithms
have smaller time complexity and better practical performance than the
prior state of the art, there is still room for improvement before we
reach the theoretical $\Omega(n)$ lower bound (which may or may not
be possible).
Further research into these questions may prove fruitful.

The techniques we develop here have applications beyond those discussed
in this paper.  For example, consider the problem of finding the
longest substring whose density matches a \emph{precise} value $\theta$.
This close relative of Problem~\ref{p-range} has
cryptographic applications \cite{boztas09-density}.  An $\OO(n)$ algorithm
is known \cite{burton11-density}, but it requires
a complex linked data structure.  By adapting and simplifying
Algorithms~\ref{a-frontier} and~\ref{a-window} for the case
$\theta_1=\theta_2=\theta$, we obtain a new $\OO(n)$ algorithm with
comparable performance and a much simpler implementation.

This last point raises the question of whether
our algorithms for the second and third problems can likewise be
simplified to use only simple array-based sorts and scans instead of the
more complex van Emde Boas trees.
Further research in this direction may yield new practical
improvements for these algorithms.


\section*{Acknowledgements}

This work was supported by the
Australian Research Council (grant number DP1094516).

\small
\bibliographystyle{amsplain}
\bibliography{pure,biology}

\vspace{1cm}
\noindent
Benjamin A.\ Burton \\
School of Mathematics and Physics, The University of Queensland \\
Brisbane QLD 4072, Australia \\
(bab@maths.uq.edu.au)

\bigskip
\noindent
Mathias Hiron \\
Izibi, 7, rue de Vaugrenier \\
89190 St Maurice R.H, France \\
(mathias.hiron@gmail.com)

\end{document}